\documentclass[preprint, showkeys,floatfix, prX]{revtex4}
\usepackage{amsmath}
\usepackage{amsthm}
\usepackage{amsbsy}
\usepackage{amssymb}
\usepackage{graphicx}
\usepackage{lineno}
\usepackage{color}
\usepackage[para,online,flushleft]{threeparttable} 
\usepackage[normalem]{ulem}

\newtheorem*{lemma}{Lemma}

\begin{document}

\title{Resilience, reactivity and variability : \\
A mathematical comparison of ecological stability measures}

\author{J-F. Arnoldi}
\email{jean-francois.arnoldi@ecoex-moulis.cnrs.fr}

\author{M. Loreau}
\author{B. Haegeman}

\affiliation{Centre for Biodiversity Theory and Modelling, Station d'\'Ecologie Exp\'erimentale du CNRS 09200 Moulis France}

\date{\today}

\begin{abstract}
 In this article, we briefly describe the RevTex 4 package for users interested in scientific publication around APS journals.
\end{abstract}

\begin{abstract}
In theoretical studies, the most commonly used measure of ecological stability is resilience: ecosystems asymptotic rate of return to equilibrium after a pulse-perturbation $-$or shock.  A complementary notion of growing popularity is reactivity: the strongest initial response to shocks.  On the other hand, empirical stability is often quantified as the inverse of temporal variability, directly estimated on data, and reflecting ecosystems response to persistent and erratic  environmental disturbances. It is unclear whether and how this empirical measure is related to resilience and reactivity. Here, we establish a connection by introducing  two variability-based stability measures belonging to the theoretical realm of resilience and reactivity. We call them intrinsic, stochastic and deterministic invariability; respectively defined as the inverse of the strongest stationary response to white-noise and to single-frequency perturbations. We prove that they predict  ecosystems worst response to broad classes of disturbances, including realistic models of environmental fluctuations. We show that they are intermediate measures between resilience and reactivity and that, although defined with respect to persistent perturbations, they can be related to the whole transient regime following a shock, making them more integrative notions than reactivity and resilience. We argue that  invariability measures constitute a stepping stone, and discuss the challenges ahead to further unify theoretical and empirical approaches to stability.
\end{abstract}


\maketitle


\section{Introduction}

What determines the stability of ecosystems has been a driving question throughout the history of ecology \citep{may_stability_1973b, pimm_complexity_1984, mccann_diversity_2000, loreau_biodiversity_2013}.  Numerous hypotheses have been proposed, explored theoretically and tested empirically.  However, the preliminary question of how to quantify stability has received less attention.  Many measures of ecological stability exist, but the choice between them is often made on purely pragmatic grounds. As a consequence, results of stability studies are often difficult to compare, because it is not clear how much these results depend on the specific choice of stability measure. In this context, clarifying the relationships and the differences between measures would be very useful.

When attempting such a clarification, one is easily overwhelmed by the vast range of regularly used stability measures.  Therefore, we start by restricting the setting in which we consider the problem of quantifying ecological stability.  First of all, we limit our attention to ecological systems whose dynamics tend to an equilibrium point.  Although it might be restrictive from an empirical viewpoint, this assumption is common in theoretical studies of ecological stability \citep{neutel_stability_2002, rooney_structural_2006, thebault_stability_2010, allesina_stability_2012}.  Indeed, this assumption allows to introduce a substantial simplification.  By focusing on the dynamics close the equilibrium point, the system can be linearized.  We assume that the equilibrium point is stable, that is, every trajectories of the linear system eventually reaches the equilibrium point.  We are then interested in quantifying the degree of stability of the linear dynamics.

Even in the simple setting of linear dynamics in the vicinity of an equilibrium point, there is a multitude of stability measures.  Typically, these measures are based on the system response to a particular perturbation (Fig.~\ref{fig:stab_measures}).  The larger the intensity or the duration of the response, the less stable the system.  The classical stability theory \citep{may_stability_1973b} is largely based on the concept of \textbf{asymptotic resilience} $\mathcal{R}_{\infty}$.  It is defined as the  asymptotic ($t\to\infty$) rate of return to equilibrium after a displacement.  The displacement does not have to decay at this asymptotic rate right away.  It might even be amplified before eventually approaching equilibrium, as captured by the notion of reactivity: the strongest initial ($t=0$) amplification of a displacement \citep{neubert_alternatives_1997}. To deduce a measure of stability, we simply define \textbf{initial resilience} $\mathcal{R}_{0}$ as the opposite of reactivity (i.e., same absolute value but opposite sign). Both resilience measures are exclusively determined by the system intrinsic dynamics.  

On the other hand, most empirical studies quantify stability as the inverse of temporal variability, directly estimated on time-series data. \citep{tilman_biodiversity_2006, jiang_different_2009, campbell_experimental_2011, donohue_dimensionality_2013}.  Although theoretical studies have also considered stability measures based on variability \citep{ives_stability_1999, lehman_biodiversity_2000, loreau_species_2008}, the link with resilience is not obvious. Indeed, in contrast with resilience, variability is caused by persistent perturbations, depends on the direction and intensity of these perturbations, and on the ecosystem variable that is observed, such as total biomass. 

As a first step in attempting to bridge the gap between empirical and theoretical measures, we define two theoretical measures of \emph{in}variability (Fig.~\ref{fig:stab_measures}):
\begin{itemize}
\item \textbf{Intrinsic stochastic invariability} $\mathcal{I}_\text{S}$ constructed from the stationary response of ecosystems to stochastic perturbations of zero-mean and persisting through time. A linear system that is perturbed by a white-noise signal eventually exhibits Gaussian fluctuations  \citep{arnold_stochastic_2013}.  The larger the variance of the stationary response, the less stable the system.  We use the inverse of this variance to define stochastic invariability $\mathcal{I}_\text{S}$ (but see section \ref{sec:stoch invariability} for a precise definition).  Stochastic white-noise perturbations are  popular in ecological studies as they are considered a simple model of environmental fluctuations \citep{may_stability_1973, ives_stability_1999,loreau_species_2008}.
\item \textbf{Intrinsic deterministic invariability} $\mathcal{I}_\text{D}$ constructed from the stationary response of ecosystems to zero-mean periodic perturbations that persist through time. A linear system that is perturbed by a periodic signal eventually oscillates at the same frequency as the driving signal \citep{ritger_differential_1968}.  The larger the amplitude of the stationary response, the less stable the system. We use the inverse of this amplitude to define deterministic invariability $\mathcal{I}_\text{D}$ (but see section \ref{sec:I_D} for a precise definition). Periodic perturbations have been used in ecological studies \citep{nisbet_population_1976, king_rainbow_1999}, and capture fundamental properties of linear systems.
\end{itemize}
Although defined for two very specific classes of perturbations, we show that the inverse of these two measures predicts ecosystems maximal response to much broader sets of disturbances: shocks occurring without temporal correlation for $\mathcal{I}_\text{S}$, and stationary perturbations with possibly long-term correlations for $\mathcal{I}_\text{D}$. This first result makes the two invariability measures complementary and easy to interpret.

By considering maximal responses over specific classes of disturbances, we have stripped several dependencies from the variability-based stability measures: they do no longer depend on direction and intensity of the applied perturbation, nor on a choice of observation variable.  Hence, the resulting invariability measures $\mathcal{I}_\text{S}$ and $\mathcal{I}_\text{D}$ are exclusively determined by the system intrinsic dynamics. Because the two resilience measures $\mathcal{R}_{\infty}$ and $\mathcal{R}_{0}$ are also intrinsic, the four stability measures can be compared.
We show that the following chain of inequalities holds in full generality,
\[
\mathcal{R}_{0} \leq \mathcal{I}_\text{S} \leq \mathcal{I}_\text{D} \leq \mathcal{R}_{\infty},
\]
meaning that, for any given system, initial resilience gives the lowest value of stability, whereas asymptotic resilience always attributes the highest. For systems with particular symmetry properties, the four measures coincide. However, this should not lead to the conclusion that the stability measures are essentially equivalent. In fact, we provide simple examples for which measures differ by orders of magnitude. 

Finally, we explain that, although defined with respect to persistent perturbations, invariability measures relate to the whole transient regime following a single shock.  In contrast, resilience measures only focus on specific short-term and asymptotically long-term responses, indicating that they are less integrative notions of ecological stability.

\begin{figure} 
\begin{center}
\includegraphics[scale=0.70]{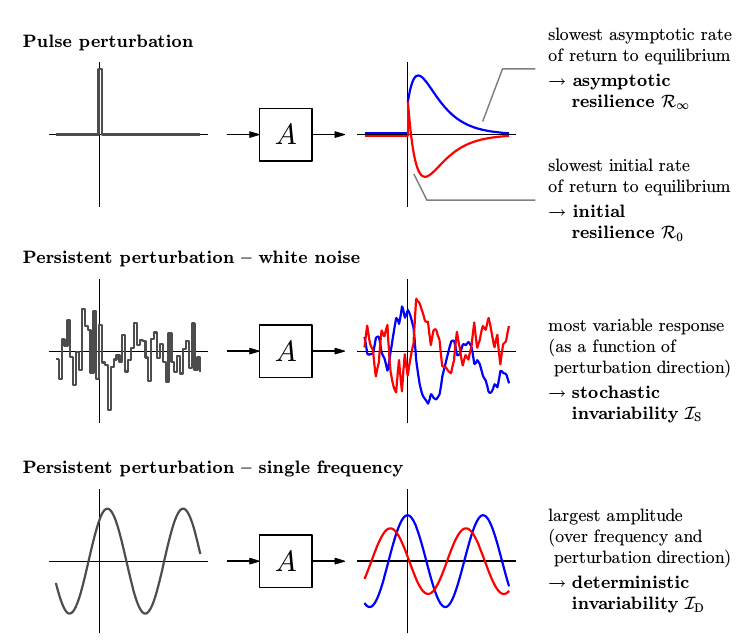}
\end{center}
\caption{
\footnotesize\label{fig:stab_measures}
\textbf{Four measures of ecological stability.}  Stability can be quantified by applying a perturbation (left graphs) to a system (here represented by community matrix $A$) and measuring its response (right graphs;  blue and red curves can be interpreted as biomass changes of two species through time).  Each stability measure we consider corresponds to the worst-case system response for a specific class of perturbations. Asymptotic resilience $\mathcal{R}_\infty$ is the slowest asymptotic rate of return to equilibrium after a pulse perturbation.  Initial resilience $\mathcal{R}_0$ is the slowest initial recovery rate.  Intrinsic stochastic invariability $\mathcal{I}_\text{S}$ is inversely proportional to the variance of the maximal response to white-noise perturbations.  Intrinsic deterministic invariability $\mathcal{I}_\text{D}$ is the inverse of the amplitude of the maximal response to single-frequency perturbations. In this paper we show that these four measures are comparable, despite the different classes of perturbations considered.
}
\end{figure} 

\section{Resilience measures}\label{sec:resiliences}

Before introducing invariability, we first describe the theoretical setting of intrinsic stability measures. We give the definitions of the classical notions of resilience (initial and asymptotic) and comment on some basic properties. We refer to Appendix \ref{sec:Notations} for the mathematical notations used throughout the paper.

Consider a non-linear dynamical system in continuous time.  It may describe, for example, a spatially structured population, a competitive community, species interacting in a food web, or abiotic and biotic flows in an ecosystem model. For convenience of speech, we shall use the terminology of a community of interacting species.  In this case, the dynamical variables correspond to species abundances or biomass, and the dynamical system describes how these abundances or biomass change over time through species interactions.
We assume there are $S$ dynamical variables, and represent these variables as a vector $\boldsymbol{N}(t)$.  The dynamical system is described by a set of coupled differential equations,
$d\boldsymbol{N}/dt = \boldsymbol{f}\left(\boldsymbol{N}\right)$.
We assume these equations admit an equilibrium point $\boldsymbol{N}_{*}$, so that
$\boldsymbol{f}\left(\boldsymbol{N}_*\right) = 0$.
The local dynamics in the vicinity of $\boldsymbol{N}_{*}$ are characterized by a matrix  
$A = D\boldsymbol{f}\left(\boldsymbol{N}_{*}\right)$,
the Jacobian of the dynamical equations evaluated at the equilibrium.  For interacting species, this matrix is called community matrix.  Denoting by $\boldsymbol{x}(t) = \boldsymbol{N}(t)-\boldsymbol{N}_{*}$ the displacement from equilibrium, the local dynamics are well approximated by a linear dynamical system: 
\begin{equation}
d\boldsymbol{x}/dt = A\boldsymbol{x}. \label{eq:lin-system}
\end{equation}
$\boldsymbol{N}_*$ is locally stable if and only if all eigenvalues of $A$ have negative real part.
Stability measures quantify the degree of stability of an equilibrium. The most common such measure is asymptotic resilience, that we now describe.

\subsection{Asymptotic resilience} 

The term resilience is used with different meanings in the ecological literature.  We use resilience as the rate of return to equilibrium, as is common in many studies of ecological stability \citep{pimm_balance_1991}.  In contrast, the definition of Holling \citep{holling_resilience_1973} is based on the size of the basin of attraction of the equilibrium.  While the latter notion is a characteristic of the non-linear dynamics, in this paper we only focus on local stability properties, encoded in the linear system (\ref{eq:lin-system}). We thus assume that under perturbations the dynamical variables remain within the basin of attraction.

Asymptotic resilience quantifies local stability as the long-term rate of return to equilibrium. 
Let us assume that at time $t=0$ a shock displaces the system to $\boldsymbol{x}(0) = \boldsymbol{x}_{0}$.  In the linear approximation, the relative abundances $\boldsymbol{x}$ evolve according to (\ref{eq:lin-system}), the solution of which is given by $\boldsymbol{x}(t) = e^{tA}\boldsymbol{x}_{0}$.  If the equilibrium is stable, any trajectory eventually leads back to it.  Using the norm $\|\boldsymbol{x}\|$ to measure Euclidean distance in phase space, the asymptotic rate of return to equilibrium reads
\[
 - \lim_{t\to\infty} \frac{1}{t} \ln \left\Vert \boldsymbol{x}(t) \right\Vert
 = - \lim_{t\to\infty} \frac{1}{t} \ln \left\Vert e^{tA}\boldsymbol{x}_{0} \right\Vert.
\]
This expression depends on the initial displacement $\boldsymbol{x}_{0}$.  To get an intrinsic stability measure, i.e., a measure that depends only on the community matrix $A$, we consider the slowest asymptotic rate of return over all initial displacements $\boldsymbol{x}_{0}$:
\[
 \mathcal{R}_{\infty}
 = \inf_{||\boldsymbol{x}_0||=1}
   \Big( -\lim_{t\to\infty} \frac{1}{t} \ln\left\Vert e^{tA}\boldsymbol{x}_0 \right\Vert \Big)
 = -\lim_{t\to\infty} \frac{1}{t} \ln\left\Vert e^{tA} \right\Vert.
\]
This equation defines an intrinsic stability measure, called asymptotic resilience. The faster the system returns to equilibrium, the more stable it is. In fact, trajectories will generically converges to the direction spanned by the eigenvector associated to the eigenvalue with largest real part, $\lambda_\text{dom}(A)$, which limits the return to equilibrium (Fig.~\ref{fig:resilience}). It follows that asymptotic resilience can be computed from this \emph{dominant} eigenvalue, $\lambda_\text{dom}(A)$, as
\begin{equation}
 \mathcal{R}_{\infty} = -\Re\big(\lambda_{\mathrm{dom}}(A)\big), \label{eq:resilience}
\end{equation}
(where $\Re(\lambda)$ is the real part of the complex number $\lambda$). If $\mathcal{R}_\infty$ is negative, some trajectories indefinitely move away.  Hence, $\mathcal{R}_\infty$ must be positive for the equilibrium to be stable. We shall sometimes refer to the eigenvector associated to $\lambda_\text{dom}$ as \emph{slow}, or \emph{dominant}, eigenvector spanning the direction of slowest return to equilibrium, towards which most trajectories converge to (note that in discrete-time dynamics, it is the eigenvalue with maximal modulus, and the associated eigenvector, that asymptotically dominate the dynamics).

The definition of asymptotic resilience is illustrated in Fig.~\ref{fig:resilience}.  For a community of $S=2$ species, we plot three trajectories in the plane $(x_1,x_2)$ (left panel).  The three trajectories have different initial conditions, corresponding to different initial displacements.  After a sufficiently long time, the distance to equilibrium decays at a fixed exponential rate (right panel; note the logarithmic scale on the $y$-axis).  This rate is the same for the three trajectories, equal to $\mathcal{R}_\infty$.

Asymptotic resilience is the most commonly used stability measure in theoretical ecology \citep{may_stability_1973b, pimm_balance_1991, neutel_stability_2002, rooney_structural_2006, thebault_stability_2010}. Note that the inverse of $\mathcal{R}_\infty$ has the dimension of time,  which is often interpreted as a characteristic return time to equilibrium.

\begin{figure} 

\begin{center}
\includegraphics[scale=0.70]{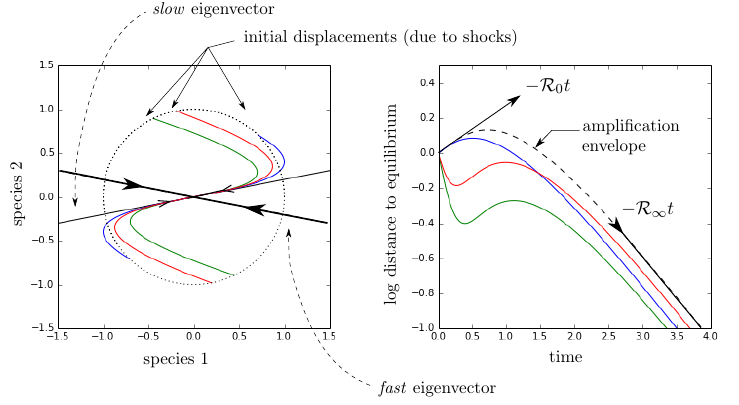}
\end{center}

\newsavebox{\smlmata} 
\savebox{\smlmata}{$A = \big( \begin{smallmatrix} -1 & 2.5 \\ 0.1 & -1 \end{smallmatrix} \big)$}
\caption{
\footnotesize\label{fig:resilience}
\textbf{Definition of asymptotic and initial resilience.}  The community matrix \usebox{\smlmata} models a mutualistic community with asymmetric interactions ($A_{12} \neq A_{21}$), near equilibrium. We have $\mathcal{R}_\infty =0.5$ and  $\mathcal{R}_0 =-0.3$, indicating that the system is reactive.  We show three trajectories (green, red, blue) starting at unit distance from the equilibrium (and their mirror image).  They represent the system response to various normalized shocks.  Left panel: plot in phase plane $(x_1,x_2)$, with the eigenvectors represented in black.  Right panel: plot of $\|\boldsymbol{x}(t)\|$ with logarithmic scale on $y$-axis. The dashed curve represents the amplification envelope, meaning the envelope of the distance to equilibrium of all trajectories starting at distance one. It is computed as the spectral norm of $e^{tA}$ (in log scale on right panel).  Asymptotic resilience $\mathcal{R}_\infty$ is the slowest asymptotic rate of return (slope for large time in right panel). Note that a displacement along the fast direction (a non-generic shock) would present a steeper asymptotic slope, corresponding to the real part of the sub-dominant eigenvalue.  Initial resilience $\mathcal{R}_0$ is the slowest initial rate of return to equilibrium (opposite of largest slope at $t=0$ in right panel).  Initial resilience can be negative, as in the example shown here, meaning that there exist trajectories (for example, the blue one) for which the initial displacement is amplified.
}
\end{figure} 

\subsection{Initial resilience} 

Asymptotic resilience characterizes the long-term response to a single shock.  However, as illustrated in Fig.~\ref{fig:resilience}, it is not necessarily related to the short-term response.  In particular, not all displacements instantly decay at the same rate.  Some displacements can even grow before eventually decaying.  When such displacements exist, the system is said to be reactive.  In \citep{neubert_alternatives_1997} reactivity is defined as the strongest initial amplification of an instantaneous displacement. We define initial resilience $\mathcal{R}_{0}$ as the opposite of reactivity, that is:
\begin{equation}
 \mathcal{R}_{0}
 = \inf_{||\boldsymbol{x}_0||=1}
   \bigg( -\frac{d}{dt} \left\Vert e^{tA}\boldsymbol{x}_0 \right\Vert \bigg|_{t=0} \bigg)
 = -\frac{d}{dt}\left\Vert e^{tA}\right\Vert \bigg|_{t=0}. \label{eq:reactivity}
\end{equation}
Initial resilience is positive when the system is non-reactive.  In this case, the larger $\mathcal{R}_{0}$, the faster the system initially evolves towards equilibrium, the more stable the system \citep{tang_reactivity_2014,suweis_resilience_2015}.  As for asymptotic resilience, $\mathcal{R}_{0}$ is an intrinsic stability measure, i.e., it depends only on the community matrix $A$.  It can be computed as the opposite of the dominant eigenvalue of the symmetric part $(A+A^\top)/2$ of $A$ ($A^{\top}$ is the transpose of $A$):
\begin{equation}
 \mathcal{R}_{0}
 = -\frac{1}{2}\lambda_{\mathrm{dom}}\left(A+A^{\top}\right). \label{eq:reactivity_spect}
\end{equation} 
The definition of initial resilience is illustrated in Fig.~\ref{fig:resilience}.  The three trajectories have different initial amplification, as can be seen from the initial slopes of the curves in the right panel.  For one of them (shown in blue), the slope is positive, meaning that the system is reactive.  In fact, this trajectory has the largest slope of all initial displacements, so that initial resilience is equal to the opposite of this slope.

The similarity of (\ref{eq:resilience}) and (\ref{eq:reactivity_spect}) shows that asymptotic and initial resilience are equal for certain matrices $A$.  In particular, if $A$ is symmetric, i.e., if $A = A^\top$, then the symmetric part $(A+A^\top)/2 = A$, and $\mathcal{R}_{0}=\mathcal{R}_{\infty}$.  More generally, this equality holds for normal matrices satisfying $A A^\top = A^\top\!A$ \citep{trefethen_spectra_2005}.  However, non-normality does not imply that $\mathcal{R}_{0} \neq \mathcal{R}_{\infty}$ $-$see (\ref{eq:example-nonnormal}) in Appendix \ref{sec:Notations}.  In the following, we call a matrix $A$ relatively reactive if  $\mathcal{R}_{0} \neq \mathcal{R}_{\infty}$. Note that a matrix is relatively reactive if it is reactive. On the other hand, a relatively reactive matrix need not be reactive. Hence, reactivity implies relative reactivity but relative reactivity does not imply reactivity.

We give a geometric intuition about relative reactivity \citep{snyder_what_2010}.  Normal matrices, which are not relatively reactive, are characterized by the property of having orthogonal eigenvectors. One can think of relative reactivity as being caused by the non-orthogonality of the eigenvectors.  This is visible in the left panel of Fig.~\ref{fig:resilience}, representing trajectories in the plane $(x_1,x_2)$.  Because the two eigenvectors are close to being collinear, some trajectories are dragged along the ``fast direction'' (associated to the non-dominant eigenvalue).  By doing so, these trajectories move away from the equilibrium while converging to the ``slow direction'' (associated to the dominant eigenvalue).

By construction, initial and asymptotic resilience are two extreme characteristics of the system recovery regime from a shock (pulse perturbation). The whole transient leading back to equilibrium cannot be expected, in general, to be fully described by the two measures of resilience. This suggests that there is room for intermediate measures of stability, taking into account the integrality of the transient.  As we shall see in the following sections, measures of temporal invariability do just that.

\section{Intrinsic stochastic invariability}\label{sec:stoch invariability}

Dynamical stability relates to the ability of a system to absorb perturbations.  To define resilience, we considered single shocks (or pulse perturbations), but these are only one type of disturbances that can be applied to the system. In fact, a simple way to model fluctuations observed on time series data, is to see them as the effect of persistent environmental disturbances. In this approach, the stable equilibrium of (\ref{eq:lin-system}) is replaced by the stationary response to those environmental perturbations. To define intrinsic stochastic invariability $\mathcal{I}_\text{S}$, we consider a specific class of stochastic disturbances, namely, white-noise perturbations, assuming that the environment fluctuates randomly and without memory.  

Mathematically, white noise is described as the ``derivative'' of Brownian motion, the continuous-time version of a random walk. To construct a Brownian motion, it is convenient to consider infinitesimal time steps, $t_k = k \delta t \rightarrow  t_{k+1} = (k+1) \delta t$, of length $\delta t$.  At each time $t_k$, a displacement is drawn from a Gaussian distribution of zero mean and variance $\delta t$.  In the continuous-time limit $\delta t \to 0$, this defines a Brownian motion $W(t)$. One defines its derivative $\xi(t)$ as the stochastic signal satisfying $W(t)= \int_{0}^{t} \xi(s) ds$, which is often written as $dW(t)=\xi(t)dt$.  The signal $\xi(t)$ is called white noise, because all frequency components have the same expected value \citep{kampen_stochastic_1997}.

We apply this type of perturbation to system (\ref{eq:lin-system}), assuming that $R$ environmental factors act on the community.  These factors $r=1,...,R$ are modeled by mutually independent white-noise signals $dW_r(t)$. The effect of environmental factor~$r$ on species~$i$ is described by a coefficient $T_{ir}$.  Explicitly, writing $X_{i}(t) = N_{i}(t)-N_{i}^{*}$, the dynamics read $dX_{i} = \sum_{j=1}^{S} A_{ij} X_{j}(t)\,dt + \sum_{k=1}^{R} T_{ik}\,dW_{k}(t)$.  Using $\boldsymbol{X} = \left(X_{1},\ldots,X_\text{S}\right)^\top$, they take the compact vector form
\begin{equation}
 d\boldsymbol{X} = A\boldsymbol{X}\,dt + T\,d\boldsymbol{W}(t),
 \label{eq:stochastic-cont}
\end{equation}
with $\boldsymbol{W} = \left(W_{1},\ldots,W_{R}\right)^\top$ a collection of independent Brownian motions. Note that species abundances $X_{i}$ must now be seen as random variables.

We focus on the stationary state $\boldsymbol{X}_{*}$ of (\ref{eq:stochastic-cont}). It has Gaussian distribution centered at the equilibrium point.  The associated stationary covariance matrix $C_{*}=\mathbb{E}\left(\boldsymbol{X}_*{\boldsymbol{X}_*}^{\!\!\top}\right)$ is the solution of the Lyapunov equation \citep{arnold_stochastic_2013}, $\hat{A}\left(C_{*}\right) + \Sigma = 0$,  with $\Sigma = T\,T^\top$ and where the operator $\hat{A}$ acts on any matrix $C$ as $\hat{A}(C) = AC+CA^\top$.  With these notations, the stationary covariance matrix reads 
\begin{equation}
 C_{*} = -\hat{A}^{-1}(\Sigma). \label{eq:expression-of-C}
\end{equation}

As for the deterministic approach, to construct an intrinsic stability measure, we seek for the perturbation that will generate the largest response.  Concretely, we look for the perturbation covariance matrix $\Sigma$ that maximizes the norm of the response covariance matrix $C_{*}$.  There are many ways to assign a norm to a matrix. For our purposes, the most convenient choice turns out to be the Frobenius norm 
$\left\Vert \Sigma\right\Vert_\text{F} = \sqrt{\mbox{Tr}\,(\Sigma^{\top}\Sigma)}$, which amounts to viewing a matrix as a vector and taking its Euclidian norm (but see Appendix  \ref{sec:harte} for a different choice).  We then define stochastic variability with respect to the Frobenius norm  as the largest stationary covariance matrix over all normalized perturbations:
\begin{equation}
 \mathcal{V}_\text{S}
 = \sup_{\Sigma\geq 0;\,\left\Vert\Sigma\right\Vert_\text{F}=1}
   \left\Vert -\hat{A}^{-1}(\Sigma) \right\Vert_\text{F}. \label{eq:VS}
\end{equation}
Finally, we define intrinsic stochastic invariability $\mathcal{I}_\text{S}$ as $\mathcal{I}_\text{S} = 1/(2 \mathcal{V}_\text{S})$.  The use of the arbitrary factor $1/2$ in this definition will become clear below.

It turns out that the supremum in (\ref{eq:VS})  without the restriction $\Sigma\geq 0$, i.e., without requiring that $\Sigma$ is a covariance matrix, gives the same result \citep{watrous_notes_2005}.  Hence,
\begin{equation}
 \mathcal{V}_\text{S}
 = \sup_{\left\Vert\Sigma\right\Vert_\text{F}=1}
   \left\Vert -\hat{A}^{-1}(\Sigma) \right\Vert_\text{F}
 = \left\Vert \hat{A}^{-1}\right\Vert, \label{eq:V_S_spectral_norm}
\end{equation}
where the norm in the last expression is the spectral norm on the space of linear operators (cf. Appendix \ref{sec:Notations}). This gives an efficient way to evaluate stochastic invariability.  Indeed, one can see $\hat{A}$ as a larger matrix $A\otimes\mathbb{I}+\mathbb{I}\otimes A$, where $\mathbb{I}$ is the identity matrix and $\otimes$ stands for the tensor, or Kronecker, product.  To compute (\ref{eq:V_S_spectral_norm}), it suffices to evaluate the spectral norm of the inverse of $A\otimes\mathbb{I}+\mathbb{I}\otimes A$. The definition of $\mathcal{I}_\text{S}$ is illustrated on Fig.~\ref{fig:stochastic_variability}.

Stochastic invariability is defined with respect to white-noise perturbations.  However, we can see white noise as a specific representative of a broad class of disturbances that yield the same definition of invariability.  It is the set of  uncorrelated shocks constructed as random sequences of instantaneous displacements occurring randomly in time.  We make this claim precise in Appendix \ref{sec:shocks}.  This shows that stochastic variability can be interpreted more generally as the maximal system response to a persistent sequence of shocks, either of infinitesimal intensity but occurring at all times, or of finite intensity but occurring at random instants.  The latter  can be more appropriate to describe certain ecological perturbations, such as drought events, wildfires or disease outbreaks.

\begin{figure} 
\begin{center}
\includegraphics[scale=0.70]{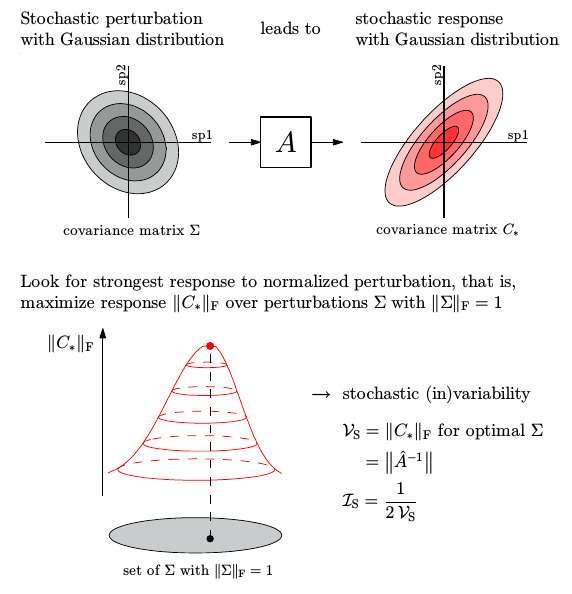}
\end{center}
\caption{
\footnotesize\label{fig:stochastic_variability}
\textbf{Definition of intrinsic stochastic invariability.}  (Top) White noise is applied to the system. It can be seen as a continuous successions of normally distributed (infinitesimal) shocks, characterized by a covariance matrix $\Sigma$. The response of the system to white noise is continuous, and  normally distributed in phase plane, with covariance matrix $C_*=-\hat{A}^{-1}(\Sigma)$. The variability of the response is measured as the Frobenius norm of this matrix.  (Bottom) To get an intrinsic measure, we look for the worst-case scenario, i.e., the input matrix $\Sigma$ generating the maximal variability.  However, we show that  the maximal response can be computed without having to solve an optimization problem. To get stochastic variability $\mathcal{V}_\text{S}$, it suffices to compute the spectral norm of $\hat{A}^{-1}$.  Stochastic invariability $\mathcal{I}_\text{S}$ is then defined as half of the inverse of $\mathcal{V}_\text{S}$.
}
\end{figure} 

\section{Intrinsic deterministic invariability}\label{sec:I_D}

In the previous section, and as if often done in theoretical studies, we modeled environmental perturbations as uncorrelated shocks. We now assume the converse, that is, we suppose the environment to be fully correlated in time. As extreme representatives of such disturbances we consider single-frequency periodic functions.  Based on this type of perturbations, we construct our last stability measure: intrinsic deterministic invariability $\mathcal{I}_\text{D}$.
We introduce deterministic environmental fluctuations $\boldsymbol{f}(t)$ in the linear dynamical as
$d\boldsymbol{x}/dt = A\boldsymbol{x}+\boldsymbol{f}(t)$.
We assume a single-frequency periodic forcing,
$\boldsymbol{f}(t) = \Re\left(e^{i\omega t}\boldsymbol{u}\right)
 = \cos(\omega t)\Re(\boldsymbol{u}) - \sin(\omega t)\Im(\boldsymbol{u})$,
where $\omega$ is the forcing frequency, $\boldsymbol{u}$ is the direction of the perturbation, and $\Re(\boldsymbol{u})$ (resp. $\Im(\boldsymbol{u})$) stands for the real part (resp. imaginary part) of the complex vector $\boldsymbol{u}$.  The perturbed dynamical system becomes
$d\boldsymbol{x}/dt = A\boldsymbol{x}+\Re\left(e^{i\omega t}\boldsymbol{u}\right)$.
The stationary system response reads
\begin{equation}
 \boldsymbol{x}(\omega, t) = \Re\left(e^{i\omega t}\boldsymbol{v}\right)
 \qquad \text{with} \qquad
 \boldsymbol{v} = \left(i\omega-A\right)^{-1}\boldsymbol{u}. \label{eq:w}
\end{equation}
We use the norm $\|\boldsymbol{v}\|$ as a measure of the system response to the forcing.  More explicitly, $\frac{1}{2}\|\boldsymbol{v}\|^2$ is the mean square distance to equilibrium,
$\frac{1}{2}\|\boldsymbol{v}\|^2 = \lim_{T\to\infty} \frac{1}{T} \int_0^T\|\boldsymbol{x}(\omega, t)\|^2 dt$.

For a given frequency $\omega$, the largest system response over all normalized perturbation vectors $\boldsymbol{u}$ is 
\begin{equation}
 \sup_{||\boldsymbol{u}||=1} \left\Vert \left(i\omega-A\right)^{-1}\boldsymbol{u} \right\Vert
 = \left\Vert \left(i\omega-A\right)^{-1} \right\Vert, \label{freq_response}
\end{equation}
where we have used the definition of the spectral norm of a matrix (see \ref{sec:Notations}).  We call $\left\Vert \left(i\omega-A\right)^{-1} \right\Vert$ the system's frequency response.  We look for the frequency $\omega$ that maximizes the frequency response, which we call the resonant frequency.  The frequency response at the resonant frequency,
\begin{equation}
 \mathcal{V}_\text{D}
 = \sup_{\omega\in\mathbb{R}}
   \left\Vert \left(i\omega-A\right)^{-1} \right\Vert, \label{eq:variability}
\end{equation}
is an intrinsic quantity, i.e., it depends only on the community matrix $A$ and represents the maximal amplitude gain over all single-frequency periodic signals.  We call $\mathcal{V}_\text{D}$ deterministic variability.  Its inverse defines an intrinsic stability measure, $\mathcal{I}_\text{D}=1/\mathcal{V}_\text{D}$, which we call intrinsic deterministic invariability.  The definition of $\mathcal{I}_\text{D}$ is illustrated in Fig.~\ref{fig:D_variability_def}. 

Quite generally, any deterministic signal can be developed into a sum of harmonic terms, or Fourier modes, of the form $\Re\left(e^{i\omega t}\boldsymbol{u}\right)$.  In the linear approximation, the system response to this general perturbation is equal to the sum of the system response to the single-frequency modes.  Then, it follows from a convexity argument that the perturbation generating the largest system response is a single-frequency mode.  Hence, when searching for the worst deterministic forcing, it suffices to consider single-frequency perturbations, as we have done in defining deterministic invariability.

We make this argument rigorous in Appendix \ref{sec:power_gain} and extend it to a large class of stationary perturbations. We relax the deterministic and periodic assumption on the environmental forcing, allowing the perturbation to be picked at random from a set of deterministic ones, that need not be periodic or even continuous. We only require that, on average (i.e., over all possible realizations), the perturbation is null, and that, again on average, its temporal autocorrelation is finite and stationary. In the language of signal analysis, such signals are called wide-sense stationary, and their maximal temporal autocorrelation defines their power.  
When comparing the output signal (the system response) to the input signal (the perturbation), we show that deterministic variability is the maximal power gain over all such stationary signals.  

As noted by Ripa and Ives \cite{ripa_food_2003}, the effect of environmental autocorrelation can be large and unintuitive.  An important feature of deterministic invariability is its ability to encompass $-$in a single number$-$ the potentiality of such effects (as long as the system remains in a vicinity of its equilibrium).

\begin{figure} 
\begin{center}
\includegraphics[scale=0.70]{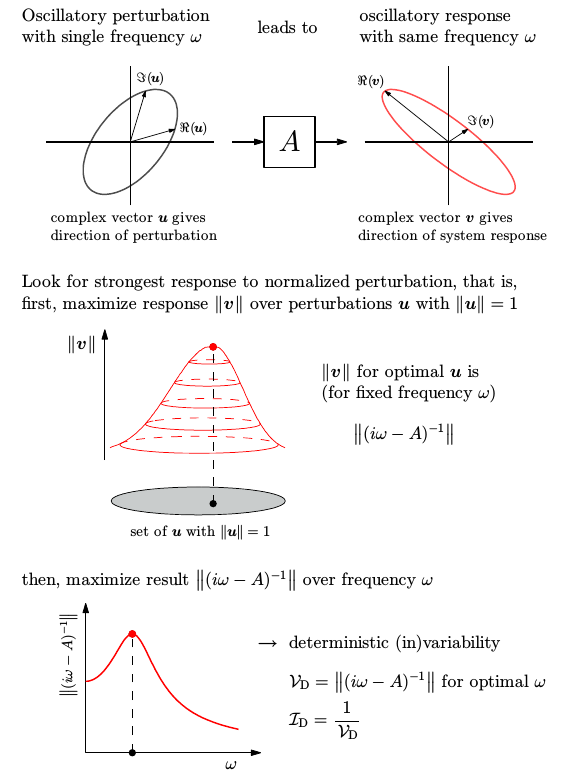}
\end{center}
\caption{
\footnotesize\label{fig:D_variability_def}
\textbf{Definition of intrinsic deterministic invariability.}  (Top) A periodic perturbation of frequency $\omega$ is applied to the system.  In phase space the perturbation defines an ellipse characterized by a complex vector $\boldsymbol{u}$. The system response oscillates a the same frequency along an ellipse characterized by the complex vector $\boldsymbol{v}=(i\omega-A)^{-1}\boldsymbol{u}$. The variability of the response is measured as the norm of $\boldsymbol{v}$.  (Middle) We then look for the maximal response over all input vectors $\boldsymbol{u}$, giving the frequency response. We get the frequency response without having to solve an optimization problem, by simply computing the spectral norm of $(i\omega-A)^{-1}$.  (Bottom) We search for the resonant frequency, giving deterministic variability as the maximal frequency response: $\mathcal{V}_\text{D}=\sup_\omega||(i\omega-A)^{-1}||$. Its inverse is deterministic invariability $\mathcal{I}_\text{D}$.
}
\end{figure} 

\section{Comparison of stability measures\label{sec:comparison of stability measures}}


In section \ref{sec:resiliences} we introduced two commonly used measures of local stability, asymptotic and initial resilience ($\mathcal{R}_{\infty}$ and $\mathcal{R}_{0}$). In sections \ref{sec:stoch invariability} and \ref{sec:I_D} we introduced stochastic and deterministic invariability ($\mathcal{I}_\text{S}$ and $\mathcal{I}_\text{D}$) and explained why they are more closely related to empirical measures $-$see table \ref{tab:summary}. Here we establish general relationships between resilience and invariability measures.

We start by considering the simplest case: one-dimensional stable equilibrium.  In the vicinity of the equilibrium, the dynamics read
$
dx/dt = -ax,
$
with $a>0$.  Note that, in this case, the matrix $A$ is scalar, $A=-a$.
To compute resilience measures $\mathcal{R}_{\infty}$ and $\mathcal{R}_{0}$, we use that, starting from $x_{0}$, the variable $x$ evolves as $x(t) = x_{0}\,e^{-at}$.  This implies that 
\[
\mathcal{R}_{\infty}=\mathcal{R}_{0}=a.                                                                                                                                                                                          
\]
To compute stochastic invariability $\mathcal{I}_\text{S}$, we must solve the Lyapunov equation (\ref{eq:expression-of-C}), with $C_{*}=\mathbb{E}(X^2_*)$ the variance of the stationary state $X_*$ associated to a stochastic forcing $\sigma^2\,dW(t)$. It simply reads    
$
(-a)C_{*}+C_{*}(-a)+\sigma^2=0,
$
so that $C_{*}=\sigma^2/(2a)$. For a normalized noise variance this gives $\mathcal{V}_\text{S} = 1/(2a)$ and 
\[
\mathcal{I}_\text{S} = 1/(2\mathcal{V}_\text{S}) = a.
\]
Finally, to compute deterministic invariability $\mathcal{I}_\text{D}$, we must solve (\ref{eq:variability}).  Here this formula takes the simple form of
\[
\mathcal{V}_\text{D}
= \sup_{\omega\in\mathbb{R}} \left|(i\omega+a)^{-1}\right|
= \sup_{\omega\in\mathbb{R}} (\omega^{2}+a^{2})^{-1/2}
= a^{-1}
\Rightarrow \mathcal{I}_\text{D} = a.
\]
Note that the maximal frequency response is attained at $\omega=0$, indicating that the perturbation with largest effect is a press perturbation, i.e., a perturbation that is constant in time.
Hence, we find that for one-dimensional dynamics, the four stability measures coincide.  This result suggest that, although at first sight their definitions are unrelated, the values of the stability measures can be expected to satisfy general relationships.  Remark that, as a corollary, we have established that the stability measures are expressed in the same units (reciprocal time), so that their values can be compared.  Also, note that this simple computation justifies the presence of the factor $1/2$ in the definition of $\mathcal{I}_\text{S}$.  Without this factor, stochastic invariability would be twice as large as the other stability measures.

Equality remains for normal matrices, such as symmetric ($A_{ij}=A_{ji}$) and skew-symmetric ($A_{ij}=-A_{ji},\;i \neq j$) matrices.  The case  $A_{ij}=A_{ji}<0$ corresponds to symmetric competitive interactions, i.e., species $i$ affects species $j$ as much as species $j$ affects species $i$.  Symmetric community matrices have been considered in previous stability studies \citep{ives_stability_1999,loreau_species_2008}.  The case $A_{ij} = A_{ji} > 0$ corresponds to symmetric mutualistic interactions, considered for instance by \cite{bastolla_architecture_2009}.
Finally, skew symmetric matrices corresponds to symmetric antagonistic interactions (e.g., predator-prey or host-parasitoid interactions)  in which the positive effect of prey species $j$ on predator species $i$ is equal (in magnitude) to the negative effect of predator species $i$ on prey species $j$.  Such matrices have been considered in theoretical studies of food webs \citep{moore_energetic_2012}. 

The equality of the stability measures in the normal case can be understood intuitively (but see \ref{sec:Notations} where we explain why normal matrices cannot be relatively reactive).  For normal matrices, the eigenvectors are orthogonal and can thus be seen as co-operating forces dragging trajectories back to equilibrium. Consequently, dynamics along the direction spanned by the ``slowest'' eigenvector (associated to asymptotic resilience) contain all of the stability limiting features, such as most reactive, most sensitive, and also largest  associated response direction. In other words, when looking for intrinsic stability measures, one can simply reduce a normal system  to a one-dimensional one along the direction spanned by its dominant eigenvector. Since stability measures coincide for one-dimensional systems, they will coincide in the normal case.  

For non-normal matrices, equality of stability measures no longer holds in general.  However, the measures are always ordered according to
\begin{equation}
\mathcal{R}_0 \leq \mathcal{I}_\text{S} \leq \mathcal{I}_\text{D} \leq \mathcal{R}_\infty,
\label{eq:dynamical_stab_ordering}
\end{equation}
as proved in \ref{sec:proof1} and illustrated in Fig.~\ref{fig:random}. It means that, for any given system, initial resilience gives the lowest value of stability, whereas asymptotic resilience always attributes the highest. Notice that the general ordering (\ref{eq:dynamical_stab_ordering}) collapses into an equality if $\mathcal{R}_0=\mathcal{R}_\infty$, i.e., whenever the community matrix is not relatively reactive. 

To illustrate the potential differences between measures, and to gain insight into the mechanisms that can cause these differences, we represent on Fig.~\ref{fig:Divergence-of-stab-measures} the behavior of stability measures for two sequences of community matrices gradually departing from normality.  On the left panel are represented the stability measures of a sequence of competitive communities near equilibrium. Species 2 has negative impact on species 1, yet species 1 has no effect on species 2 (an amensalistic interaction). As the asymmetry of the interaction grows, asymptotic resilience remains unchanged while other measures decrease. For large enough asymmetry, asymptotic resilience is one order of magnitude larger than invariability measures. On the right panel of Fig.~\ref{fig:Divergence-of-stab-measures} the matrices model a consumer-resource system near equilibrium, with the consumer depleting, for a fixed benefit, an increasing amount of resource, while increasing its tendency to return to equilibrium. In this artificial example, the stability trend along the gradient appears to be ambiguous. Indeed, as the interaction asymmetry grows, asymptotic resilience increases while other measures indicate a sharp loss of stability.  We notice in both examples that systems become reactive (i.e., $\mathcal{R}_0<0$) while remaining relatively stable with respect to other measures. 

\begin{table} 
\begin{center}
\begin{threeparttable}
\small\begin{tabular}{lcr}
Stability measure &  Interpretation & Formula \\[-3pt]
\hline \\[-12pt]

Asymptotic   & Slowest asympt. rate of return   &  $\mathcal{R}_\infty = -\Re(\lambda_\text{dom}(A))$ \tnote{(a)}                               \\[-3pt]
resilience   & to equilibrium after a shock.    &                                                                                   \\

Deterministic& Inverse of maximal response      & $\mathcal{I}_\text{D} = (\sup_\omega ||(i\omega-A)^{-1}||)^{-1} $ \tnote{(b)}                  \\[-3pt]
invariability& amplitude to periodic forcing.   &                                                                                    \\

Stochastic   & Inverse of maximal response      &  $\mathcal{I}_\text{S} = \frac{1}{2} ||-\hat{A}^{-1}||^{-1} $ \tnote{(c)}                      \\[-3pt]
invariability& variance to white-noise.         &                                                                                    \\

Initial      & Slowest initial rate of return   &  $\mathcal{R}_0 = -\frac{1}{2} \lambda_\text{dom}(A+A^\top)$ \tnote{(d)}                      \\[-3pt]
resilience   & to equilibrium after a shock.    &                                                                                    \\
\hline
\end{tabular}

\begin{tablenotes}
\footnotesize
\item[(a)] $\lambda_\text{dom}$ is the eigenvalue of community matrix $A$ with maximal real part $\Re(\lambda_\text{dom})$.\\ 
\item[(b)] $i$ is the imaginary unit and $\omega\geq 0$. $||\cdot ||$ is the spectral norm of matrices.\\
\item[(c)] $\hat{A}=A\otimes\mathbb{I}+\mathbb{I} \otimes A$ where $\mathbb{I}$ is the identity matrix; $\otimes$ is the Kronecker product.\\
\item[(d)] $A^{\top}$ is the transpose of $A$.
\end{tablenotes}
\end{threeparttable}

\end{center}
\caption{
\footnotesize\label{tab:summary}
\textbf{Computing intrinsic measures of stability}. The two resilience measures are well known, while the two invariability measures are new. All four measures are defined with respect to the worst system response over different types of perturbations. They are expressed in the same units (reciprocal time) and can be computed directly from the community matrix $A$. 
}

\end{table} 
\begin{figure} 
\begin{center}
\includegraphics[scale=0.70]{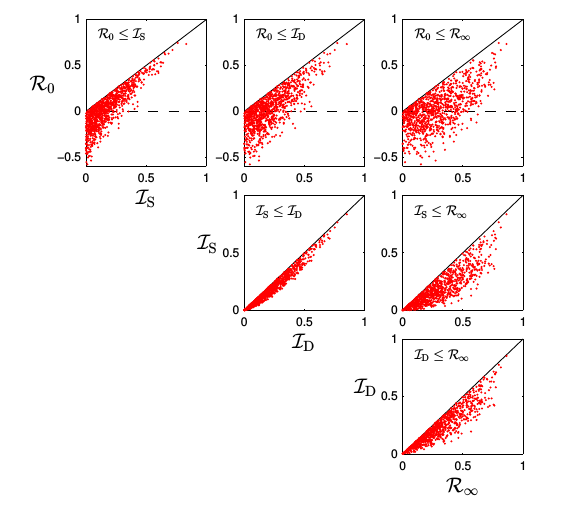}
\end{center}
\caption{
\footnotesize\label{fig:random}
\textbf{Illustration of the general stability ordering (\ref{eq:dynamical_stab_ordering})} on random matrices of dimension $S=3$.  The diagonal elements were drawn from a uniform distribution over $[-1,0]$ while off-diagonal elements were drawn from a normal distribution of mean $0$ and variance $1$. With this procedure $63\%$ of the matrices generated were stable.  We plotted all four stability measures against each other for 1000 stable matrices (each red dot corresponds to a stable matrix).  Dots lying below the full black lines correspond to relatively reactive matrices; dots lying below the dashed black line (top row panels) correspond to reactive matrices. 
}
\end{figure} 

\begin{figure} 
\begin{center}
\includegraphics[scale=0.50]{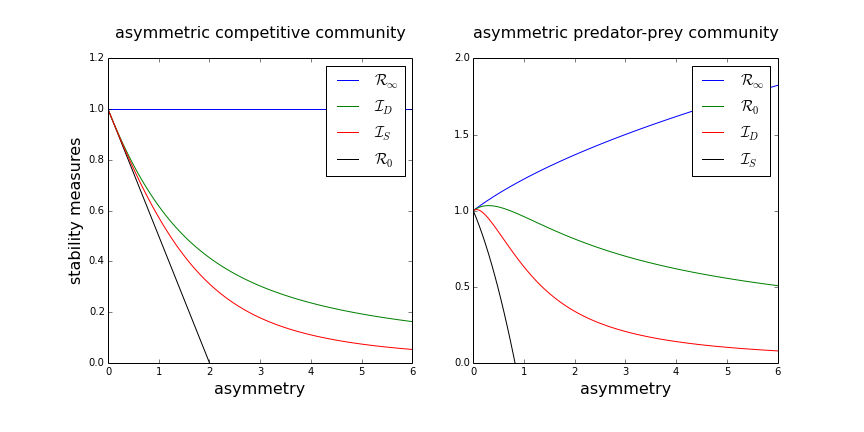}
\end{center}
\newsavebox{\smlmatb} 
\savebox{\smlmatb}{$A = \big( \begin{smallmatrix} -1 & -\rho \\ 0 & -1 \end{smallmatrix} \big)$}
\newsavebox{\smlmatc} 
\savebox{\smlmatc}{$A = \big( \begin{smallmatrix} -1 & -(1+\rho)^2 \\ 1 & -\sqrt{1+\rho} \end{smallmatrix} \big)$}
\caption{
\footnotesize\label{fig:Divergence-of-stab-measures}
\textbf{Branching of stability measures} as interaction asymmetry increases.  Left panel: the community matrix \usebox{\smlmatb} models the dynamics of two species near equilibrium, with one of them having negative impact on the other.  As interaction asymmetry $-$parametrized by $\rho\geq0-$ grows, asymptotic resilience remains unchanged while the other measures drop.  At $\rho=6$ asymptotic resilience is one order of magnitude larger than stochastic invariability.  Right panel: the matrix \usebox{\smlmatc} models a consumer-resource system near equilibrium, with the consumer depleting, for a fixed benefit, a growing  amount of resource.  As asymmetry grows, asymptotic resilience increases while other measures show a loss of stability. We observe that in both examples, the system becomes reactive ($\mathcal{R}_0<0$) at low asymmetry.
}
\end{figure} 

All measures can be expressed as characteristics of the transient regime following a shock and leading back to equilibrium. This claim might seem surprising, as invariability measures are defined with respect to persistent disturbances and not to pulse perturbations. To reveal this link, notice that an external forcing, either deterministic or stochastic, constantly pushes the system away from equilibrium and can be seen as a sequence of pulse perturbations. The system stationary response is, at each time, the sum of the responses to all past perturbations. Hence, invariability measures are sensitive to short-term responses, long-term ones, and all in between; in other words, to the whole transient regime leading back to equilibrium. The envelope of the distance to equilibrium of all trajectories (associated to all normalized shocks that can be applied to the system) defines the so-called amplification envelope (see Fig.~\ref{fig:resilience}). This suggests a link between invariability measures and the integral of the amplification envelope (see \ref{sec:harte}). By definition, initial and asymptotic resilience relate to the head and tail of this envelope. When the transient is completely determined by asymptotic resilience, stability measures coincide, but in general, neither initial nor asymptotic resilience fully characterize the transient, hence stability measures differ.

The above reasoning also sheds light on the reasons why measures are ordered according to (\ref{eq:dynamical_stab_ordering}). First of all, to understand why $\mathcal{R}_0$ is smaller than $\mathcal{R}_\infty$, it suffices to notice that the initial decay of a displacement along the dominant eigenvector of the community matrix is precisely given by asymptotic resilience. Since initial resilience corresponds to the worst-case scenario, it can only be smaller. A similar argument applies for other measures, by considering perturbations along the dominant eigenvector.  In the case of persistent disturbances, the system stationary response is, at each time, the sum of the responses to all past displacements. All these displacements are dragged towards the dominant eigenvalue, thus their absorption rate changes until reaching the decay rate predicted by $\mathcal{R}_\infty$. The resultant response is thus always smaller than if all displacements were only absorbed at the minimal initial rate $\mathcal{R}_0$, which implies that variabilities are smaller than the inverse of $\mathcal{R}_0$. This explains why invariability measures are framed by resilience measures. 

The fact that $\mathcal{I}_\text{S} \leq \mathcal{I}_\text{D}$ is less intuitive, and specific to the normalization choices we have made.  We stress that it should not be interpreted as if uncorrelated shocks generate larger variability than autocorrelated fluctuations. The white-noise normalization chosen to define stochastic variability focuses on the variance of the displacements induced by shocks and not on the variance of the signal itself, which is infinite. Yet, uncorrelated shocks generate a system response with finite variance. In terms of gain of variance uncorrelated shocks are thus far less efficient in exciting the system than autocorrelated signals.

Finally, it is worth noting that there is no biological reason why community matrices representing biological systems should be normal, or simply not relatively reactive. In fact, it has been established that many natural systems are reactive \citep{neubert_alternatives_1997, neubert_reactivity_2004, neubert_detecting_2009}. Since reactivity is a stronger condition than relative reactivity, which in turn is a stronger condition than non-normality, this suggest that most natural systems are non-normal and relatively reactive. We have explained that, in this case, stability measures can largely differ. This advocates for a more integrative approach to local stability, that does not simply focus on asymptotic resilience.

\section{Discussion}

Ecological stability theory is largely based on the response of ecosystems to single pulse perturbations, or shocks.  This corresponds to the  mathematical definition of resilience $-$either initial or asymptotic$-$ derived from the theory of linear dynamical systems \citep{may_stability_1973b, neubert_alternatives_1997}. 
Resilience measures are rarely used in empirical studies because, amongst other reasons, environmental perturbations occur all the time.  Instead, empirical stability is typically expressed as the inverse of temporal variability, directly measured on time series data. 
This has inspired theoretical studies to consider variability-based stability measures \citep{ives_stability_1999, lehman_biodiversity_2000, loreau_biodiversity_2013}, 
yet these approaches remain largely disconnected from the large body of resilience-based stability theory. Indeed,
several obstacles stand in the way to establish a clear link between empirically motivated and purely theoretical views on stability (see Table~\ref{tab:comptheory}):
\begin{itemize}
\item Variability is caused by persistent environmental disturbances, while resilience theory considers single-shock perturbations.
\item In previous studies, variability depends on the intensity and direction of environmental perturbations, whereas resilience measures do not depend on perturbation features.
\item In previous studies, variability is measured on a specific variable (like total biomass), whereas resilience measures are defined independently of a choice of observed variable.
\end{itemize}

To narrow the gap between variability- and resilience-based stability,  we focused on the fundamental discrepancy (i).  We introduced two new variability measures that originate from the same theoretical setting as resilience, surmounting discrepancies (ii) and (iii). They are constructed from the maximal response to two distinct types of persistent disturbances:  shocks occurring without temporal correlation and stationary perturbations with long-term correlations. We called them \emph{intrinsic} measures of invariability, to emphasize that they only depend on the intrinsic ecosystem dynamics, i.e., on the community matrix.

Because resilience measures are also intrinsic, stochastic and deterministic invariability allowed for a general comparison of stability measures.  In doing so, we found that invariability measures are intermediate  between initial and asymptotic resilience. We explained this result as a consequence of the fact that, although defined with respect to persistent perturbations, invariability relates to the whole transient regime following a shock, while resilience only focuses on specific short-term and asymptotically long-term responses.

\begin{table} 
\begin{center}
\small\begin{tabular}{rlccc}
 & & asympt/initial & existing theory & intrinsic \\[-2pt]
 & & resilience & of variability & invariability \\[6pt]
 \hline \\[-12pt]
 (i) & type of perturbation & pulse & persistent & persistent \\
 (ii) & depends on perturbation features & no & yes & no \\
 (iii) & associated to an observed variable & no & yes & no
\end{tabular}
\end{center}
\caption{
\footnotesize\label{tab:comptheory}
\textbf{Intrinsic measures of invariability as a stepping stone} between purely theoretical and empirically motivated notions of stability.  Different classes of stability measures are compared based on whether the system response to pulse or persistent perturbations is considered (i); whether the measures depend on the intensity and direction of the applied perturbation (ii); and whether the system response is measured on a specific variable $-$such as total biomass; (iii). 
A large part of ecological theory uses intrinsic stability measures associated to pulse perturbations (like asymptotic and initial resilience; column ``asympt/initial resilience'').  There also exists a rather disconnected theory of ecological variability, based on non-intrinsic variability measures and persistent perturbations (column ``existing theory of variability'').  In this paper we bridge these two approaches by introducing intrinsic invariability measures (column ``intrinsic invariability'').
}
\end{table} 

\begin{figure} 
\begin{center}
\includegraphics[scale=0.70]{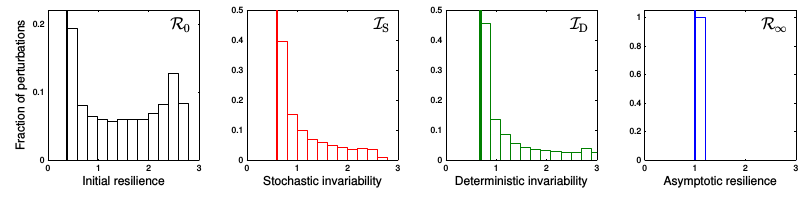}
\end{center}
\newsavebox{\smlmatd} 
\savebox{\smlmatd}{$A = \big( \begin{smallmatrix} -1 & 2 \\ 0 & -2 \end{smallmatrix} \big)$} 
\caption{
\footnotesize\label{fig:genericity}
\textbf{Worst-case vs generic perturbations.} For the community matrix \usebox{\smlmatd},  we investigate the system response to random perturbations and compare with the worst-case scenario.  In each panel, the vertical thick line represents the intrinsic stability measure, i.e., the system response corresponding to the worst-case scenario, over the defining set of perturbations associated to that measure. The histogram represents the distribution of the system response to a perturbation randomly sampled from the defining set.  For resilience measures (left- and rightmost panel), we generated 1000 initial displacements, drawn uniformly on the unit sphere around equilibrium.  For stochastic invariability (second panel), we generated 1000 random matrices $T$ of independent Gaussian variables, and constructed white-noise covariance matrices by normalizing $\Sigma=T\,T^\top$ (Wishart distribution).  For deterministic invariability (third panel), we generated 1000 random press perturbations (i.e., of frequency $\omega=0$, which is the resonant frequency of this system), uniformly distributed on the sphere. Only asymptotic resilience is generic, as all asymptotic return rates equal $\mathcal{R}_\infty$. For other measures, the worst-case response can be very different from the response to a particular perturbation.  
}
\end{figure} 

While this result establishes a fundamental link between variability and resilience, it does not make the connection with empirical and (empirically motivated) variability.  In particular, empirically measured variability depends on a specific environmental perturbation acting on the system, while intrinsic measures of invariability and resilience inform on the worst-case system response over entire sets of perturbations.  In Fig.~\ref{fig:genericity} we illustrate the fact that only asymptotic resilience represents a generic response to pulse perturbations. In other words, for all measures but asymptotic resilience, one must expect, in general, the worst-case response to be very different from the response to a particular perturbation.

This indicates that in a context where the nature and direction of environmental disturbances are expected to change,  focusing on a specific perturbation direction to assess stability can be misleading in terms of informing on the potential threats to ecosystems. 
Intrinsic measures, although not directly observable on data, thus contain important stability information that empirical measures might miss. This statement should however be taken with a note of caution: invariability and resilience measures do not always relate to realistic $-$or observable$-$ perturbation scenarios.  Indeed, environmental disturbances generally do not affect species abundances directly, and species populations respond differently, depending on their functional traits and abundances.  This will associate different intensities to different directions of perturbations, hence potentially restricting the response possibilities.  For instance, it seems natural to assume that a perturbation affecting an abundant species is stronger than a perturbation affecting a less abundant one.  We will investigate the consequences of the scaling of perturbation intensity by abundance in future work.  It is interesting to note already that, while this perturbation scaling will not affect resilience (hence the resilience of an ecosystem can potentially be determined by the response of rare $-$even unobserved$-$ species), it can qualitatively modify stability patterns as predicted by invariability, suggesting that invariability could be a more flexible stability notion than resilience. 

If invariability measures defined in this article constitute a stepping stone, the gap between theoretical and empirical stability remains far from being bridged.  In this regard, the underlying equilibrium assumption constitutes arguably the most serious obstacle.  This assumption is rarely satisfying to approach real ecological systems, which can sometimes display much more complex stationary dynamics (see \citep{beninca_species_2015} for a particularly convincing example); or can simply not be in a stationary regime, due to recent environmental change.  However, it should be noted that the framework developed in this paper can, to some extent, be generalized.  For example, intrinsic invariability measures can be transposed to discrete-time dynamical systems, which are important in their own right, but also to deal with periodic ecological dynamics in continuous time. In this case, the equilibrium assumption is not relevant, but can remain valid after making a stroboscopic section of trajectories, using the so-called Poincar\'e map \citep[chap. 11]{dieckmann_geometry_2000}. This illustrates that the results we have obtained for a restricted theoretical setting can have wider applicability. 

Ecosystems across the planet face a myriad of environmental stress, and threats. In a context of global environmental crisis, there is a glaring need for conceptual tools to better understand the complex dynamics of nature.
For near-equilibrium dynamics, we illustrated that focusing on the dominant direction of return to equilibrium can be misleading, and that a more integrative approach is possible, providing unintuitive insight on systems response to potential perturbations.
If this idealistic setting can be used in other cases such as periodic dynamics, it should also serve as a reference point to move towards more realistic ecosystem models.  The fact that there was unexploited richness in such a simple setting suggests that ecological stability theory can be significantly improved without having to resort, yet, to overly complicated mathematical formalism.

\acknowledgements{
We warmly thank Shaopeng Wang, Claire de Mazancourt and Jos\'e Montoya for numerous and fruitful discussions, that greatly helped the redaction of this article. J-F. A. is particularly grateful to Shaopeng Wang and N\'uria Galiana for their enthusiastic support and critical reading of preliminary drafts. 
This work was supported by the TULIP Laboratory of Excellence (ANR-10-LABX-41) and by the AnaEE France project (ANR-11-INBS-0001).
}

\bibliography{stabtheorbiol}

\appendix
\section{Mathematical notations\label{sec:Notations}}

In this article a vector $\boldsymbol{u}$ is seen as a column whereas its transpose $\boldsymbol{u}^{\top}$ is a row.  For complex vectors the dual is taken as  $\boldsymbol{u}^{*}=\overline{\boldsymbol{u}^{\top}}$, the conjugate transpose, so that $\boldsymbol{u}^{*}\boldsymbol{v}=\left\langle \boldsymbol{u},\boldsymbol{v}\right\rangle$ is the Hermitian scalar product (Notice that reversing the order gives $\boldsymbol{u}\boldsymbol{v}^{*}$, a rank-one matrix).  The associated (Euclidean) norm is then 
\begin{equation}
\boldsymbol{u}^{*}\boldsymbol{u}=\left\langle \boldsymbol{u},\boldsymbol{u}\right\rangle =\left\Vert \boldsymbol{u}\right\Vert ^{2}\label{eq:scalar}
\end{equation}
This norm on vectors induces a norm on matrices $B$, called the spectral norm: 
\begin{equation}
\left\Vert B\right\Vert =\sup_{\left\Vert \boldsymbol{u}\right\Vert =1}\left\Vert B\boldsymbol{u}\right\Vert \label{eq:2norm}
\end{equation}
The dominant eigenvalue of a given matrix $B$ (i.e., with largest real part) is denoted $\lambda_{\mathrm{dom}}\left(B\right)$. If $B$ is a complex matrix, its adjoint is given by $B^*$, the conjugate transpose.  In particular, it holds that $\left\Vert B\right\Vert ^{2}=\lambda_{\mathrm{dom}}\left(B^*B\right)$, which justifies the term ``spectral norm'' for $||B||$.

The space of matrices $\mathbb{C}^{S\times S}$ is endowed with an inner product structure $\left\langle A,B\right\rangle =\mbox{Tr}\left(A^{*}B\right)$, where $\mbox{Tr}$ stands for the trace, giving the sum of diagonal elements of square matrices. The Schatten norms reflect this structure as $\left\Vert B\right\Vert _{p}=\mbox{Tr}\left(\left|B\right|^{p}\right)^{\frac{1}{p}}$ with $\left|B\right|=\sqrt{B^{*}B}$.  In this article we only consider $p=1$, the trace norm, $p=2$, the Frobenius norm, compatible with the inner product, and $p=\infty$, the spectral norm.
We also endow the space of linear operators $\mathcal{L}\left(\mathbb{C}^{S\times S} \right)$ $-$acting on matrices$-$ with a norm, induced by the Frobenius norm, as done by \citep{watrous_notes_2005}: 
\begin{equation}
\forall\mathcal{B}\in\mathcal{L}\left(\mathbb{C}^{S\times S}\right),\;\left\Vert \mathcal{B}\right\Vert =\sup_{\left\Vert U\right\Vert _\text{F}=1}\left\Vert \mathcal{B}U\right\Vert _\text{F}\label{eq:lifted-shatten}
\end{equation}
An important remark is that the lifted norm $\left\Vert\mathcal{B}\right\Vert$ coincides with the spectral norm on the space of linear operators $\mathcal{L}\left(\mathbb{C}^{S\times S} \right)$. 

A matrix $A$ is said to be normal if it commutes with its adjoint $A^*$ \citep{trefethen_spectra_2005}.  Hence $A$ and $A^*$ have the same eigenvectors, associated to conjugate eigenvalues.  This implies that the dominant eigenvalue of $(A+A^*)/2$ is equal to the real part of the dominant eigenvalue of $A$.  In particular if $\mathcal{R}_X, \; (X=0,\infty),$ stands for the two resilience measures defined in the main text, we get that $\mathcal{R}_0=\mathcal{R}_\infty$. Hence normal matrices are never relatively reactive.  However, the set of relatively reactive matrices is smaller than the one of non-normal matrices, as can be proved by considering 
\begin{equation}
A=\left(\begin{array}{ccc}
  -1 & 0 & 0 \\
  0 & -2 & 0 \\
  0 & \epsilon & -2
\end{array}\right)\label{eq:example-nonnormal},
\end{equation}
this matrix is not normal for $\epsilon \neq 0$, yet not relatively reactive either, as long as $\left|\epsilon\right| \leq 2$. In this example, the eigenvector associated to the dominant eigenvalue $-1$ is orthogonal to the  subspace associated to the (degenerate) sub-dominant eigenvalue $-2$.  The non-normality of the restriction of $A$ to that subspace needs to be sufficiently pronounced to have a significant dynamical impact on the associated linear system $d\boldsymbol{x}/dt=A\boldsymbol{x}$.

\section{Uncorrelated shocks, white noise, and stochastic variability\label{sec:shocks}}

\subsection{Stochastic variability is the maximal system response to uncorrelated shocks}

We define a class of stochastic perturbations that yields the same definition of intrinsic stochastic invariability than the one associated to white noise $-$section \ref{sec:stoch invariability} in the main text.  This class consists of uncorrelated shocks, occurring at random instants. They take the form of a random sequence of pulse perturbations:
\[
\boldsymbol{\xi}_{\lambda,\Sigma}(t)=\sum_{k}\boldsymbol{u}_{k}\delta\left(t-t_{k}\right)
\]
where vectors $\boldsymbol{u}_{k}\in\mathbb{R}^{S}$ and times $t_{k}$ are independent random variables, parametrized by correlation matrix $\Sigma$ and intensity $\lambda$.  More precisely, the vectors $\boldsymbol{u}_{k}$ are independent, identically distributed variables drawn from a distribution of zero mean: $\mathbb{E} \boldsymbol{u}=0$ and correlation matrix $\Sigma=\mathbb{E}\boldsymbol{u}_{k}^{\top}\boldsymbol{u}_{k}$.  They represent the amplitude and direction of the displacement occurring at time $t_{k}$. The times $t_{k}$ are generated by a Poisson process with intensity $\lambda$.  They represent the time coordinate of perturbation events. The average number of events in a time period of length $T$ is $\lambda T$. We normalize the intensity $\lambda$ and the matrix $\Sigma$ such that $\lambda\left\Vert \Sigma\right\Vert _{\text{F}}=1$, where $\left\Vert \cdot \right\Vert _{\text{F}}$ stands for the Frobenius norm of matrices (\ref{sec:Notations}).  This can be interpreted as a trade-off between frequency of events and amplitude of the associated displacement.  The effect of such disturbances on a
community near equilibrium is modeled through the following linear dynamical system
\[
d\boldsymbol{x}/dt=A\boldsymbol{x}+\boldsymbol{\xi}_{\lambda,\Sigma}(t).
\]
The stationary response of the community can then be written explicitly as
\[
\boldsymbol{x}(t)=\sum_{k|t_{k}<t}e^{(t-t_{k})A}\boldsymbol{u}_{k}
\]
The mean system response is zero, and the associated covariance matrix reads
\[
C_{*}=\mathbb{E}_{\left\{ \boldsymbol{u}_{k},t_{k}\right\} }\boldsymbol{x}(t)\boldsymbol{x}(t)^{\top}=\mathbb{E}_{\left\{ t_{k}\right\} }\sum_{k|t_{k}<t}e^{(t-t_{k})A}\Sigma e^{(t-t_{k})A^{\top}}
\]
\[
=\int_{_{-\infty}}^{t}e^{(t-s)A}\Sigma e^{(t-s)A^{\top}}\lambda ds=\int_{0}^{\infty}e^{sA}\lambda\Sigma e^{sA^{\top}}ds
\]
where, in the last term, we recognize the unique solution to the Lyapunov equation $AC_{*}+C_{*}A^{\top}=\lambda\Sigma$ \citep{arnold_stochastic_2013}, so that $C_{*}=\hat{A}^{-1}\left(\lambda\Sigma\right)$, where $\hat{A}(\cdot)=A\cdot+\cdot A^{\top}$ is the lifted linear operator defined in section \ref{sec:stoch invariability}. Hence, using the Frobenius norm, the maximal response over all normalized uncorrelated shocks, is
\[
\sup_{\lambda\left\Vert \Sigma\right\Vert _{\text{F}}=1} \left\Vert \hat{A}^{-1}\left(\lambda\Sigma\right)\right\Vert _{\text{F}}=\left\Vert \hat{A}^{-1}\right\Vert =\mathcal{V}_\text{S}
\]
that is, stochastic variability as defined in section \ref{sec:stoch invariability}.  Stochastic variability can thus be interpreted more generally as the maximal system response to uncorrelated disturbances, either infinitesimal
shocks occurring at all times (that is, white noise) or finite shocks occurring at random instants (that is, the above class of uncorrelated shocks).  In fact, we now prove that white noise is a specific representative of the class of uncorrelated shocks.

\subsection{White noise as a limit case of uncorrelated shocks}

In section \ref{sec:stoch invariability} we defined white noise as the derivative of the Brownian motion. We shall use this definition to prove that white noise is a limit case of uncorrelated shocks. For the sake of simplicity, we limit our attention to the one-dimensional case.

Consider time instants $t_k$ generated by a Poisson process with rate $\lambda$.  Consider independent random variables $u_k$ with identical distribution.  This distribution is not necessarily normal;  we only assume that it has zero mean and finite variance $\sigma^2$.  The associated uncorrelated shocks read
\[
 \xi_{\lambda,\sigma^2}(t) = \sum_k u_k\,\delta(t-t_k)
\]
We claim that the joint limit $\lambda\to\infty$, $\sigma^2\to 0$ with $\lambda \sigma^2 = 1$ yields the one-dimensional white-noise signal.  We prove this by defining, for $s_1<s_2$, the random variable 
\[
 W_{\lambda,\sigma^2}(s_1,s_2) = \int_{s_1}^{s_2} \xi_{\lambda,\sigma^2}(s)\,ds = \sum_{k | s_1<t_k<s_2} u_k
\]
a sum of independent and identically distributed random variables.  The number of terms in the sum is Poisson distributed with mean $\lambda (s_2-s_1)$.  The mean of $W_{\lambda,\sigma^2}(s_1,s_2)$ is zero and its variance is $\lambda (s_2-s_1)\,\sigma^2$. Moreover, $W_{\lambda,\sigma^2}(s_1,s_2)$ is independent of $W_{\lambda,\sigma^2}(s_3,s_4)$ for $s_1<s_2<s_3<s_4$.

For large $\lambda$ the number of terms in the sum is typically large, and we can apply a generalized central limit theorem \citep{blum_central_1963, renyi_central_1963}.  We find that 
\[
\lim_{\lambda \rightarrow \infty} W_{\lambda,\lambda^{-1}}(s_1,s_2)  = \mathcal{N}(0,s_2-s_1)
\] 
where $\mathcal{N}$ stands for the normal distribution. This indicates that $W_{\lambda,\lambda^{-1}}(0,t)$ converges to the Brownian motion and thus that $\xi_{\lambda,\lambda^{-1}}(t)$ converges to white noise.

\section{Deterministic variability and stationary perturbations\label{sec:power_gain}}

Consider $\boldsymbol{u}:\Omega\times\mathbb{R}_{t}\rightarrow\mathbb{C}^{S}$
a (wide-sense) stationary signal, defined on a probability space $\Omega\ni\lambda$. This is the input.
We assume zero mean: 
$
\mathbb{E}\boldsymbol{u}(t)=\int\boldsymbol{u}(\lambda,t)d\lambda=0
$
and finite power at any given time:
\[
\mathbb{E}\left\Vert \boldsymbol{u}(t)\right\Vert ^{2}=\int\left\Vert \boldsymbol{u}(\lambda,t)\boldsymbol{u}(\lambda,t)^{*}\right\Vert _{1}d\lambda
\]
Here $\left\Vert \cdot\right\Vert _{p}$ stands for the Schatten norm of matrices ($p=1$ is trace norm, $p=2$ is Frobenius norm and $p=\infty$ is spectral norm).  Writing 
\[
\gamma_{\text{in}}\left(t,\tau\right)=\mathbb{E}\boldsymbol{u}(t)\boldsymbol{u}(t-\tau)^{*}
\]
for the signal autocorrelation matrix, we see that the power of the input is given by $\left\Vert \gamma_{\text{in}}(t,0)\right\Vert _{1}$.  Wide-sense stationarity means that the autocorrelation is independent of $t\in\mathbb{R}$; we therefore drop that variable from now on.  Consider now, for any realization $\lambda\in\Omega$, the dynamical system
\[
d\boldsymbol{x}=A\boldsymbol{x}dt+\boldsymbol{u}(\lambda,t)dt
\]
where $A$ is a stable matrix. The system's stationary state reads, for any $\lambda\in\Omega$
\[
\boldsymbol{x}(t,\lambda)=\int_{-\infty}^{t}e^{(t-s)A}\boldsymbol{u}(\lambda,s)ds=\int_{t}^{\infty}e^{(s-t)A}\boldsymbol{u}(\lambda,-s)ds
\]
This defines (one realization) of the output signal. We wish to estimate the power of the output. By definition it is given by the trace norm of the autocorrelation matrix 
$
\gamma_{\text{out}}(0)=\mathbb{E}\boldsymbol{x}(0)\boldsymbol{x}(0)^{*}
$
Precisely, we prove the following (sharp) upper bound on that power:
\[
\left\Vert \gamma_{\text{out}}\left(0\right)\right\Vert _{1}\leq\mathcal{V}_\text{D}^{2}\,\left\Vert \gamma_{\text{in}}\left(0\right)\right\Vert _{1}
\]
showing that  $\mathcal{V}_\text{D}$, deterministic variability, is the maximal power gain that the system can generate.  To see this, notice first that 
\begin{align*}
 \gamma_{\text{out}}\left(0\right)
 &= \int_{0}^{\infty} \int_{0}^{\infty} 
    e^{s_{1}A} \mathbb{E}\boldsymbol{u}(-s_{1})\boldsymbol{u}(-s_{2})^{*} e^{s_{2}A^{*}} ds_{1}ds_{2} \\
 &= \int_{0}^{\infty} \int_{0}^{\infty}
    e^{s_{1}A} \gamma_{\text{in}}\left(s_{1}-s_{2}\right) e^{s_{2}A^{*}} ds_{1}ds_{2}
\end{align*}
From the Wiener-Khinchin-Einstein theorem \citep{wiener_generalized_1930}, $\gamma_{\text{in}}\left(\tau\right)$
can be decomposed with respect to its power spectral density $d\hat{\gamma}_{\text{in}}\left(\omega\right)$
as
\[
\gamma_{\text{in}}\left(\tau\right)=\int_{\mathbb{R}}e^{-i\omega\tau}d\hat{\gamma}_{\text{in}}\left(\omega\right)
\]
where, defining the truncated Fourier transform $\hat{\boldsymbol{u}}_{T}(\lambda,\omega)=\left(2\pi\right)^{-1}\int_{0}^{T}\boldsymbol{u}(\lambda,t)e^{i\omega t}d\omega$, the power spectral density can be constructed as 
\[
 d\hat{\gamma}_{\text{in}}\left(\omega\right)
 = \lim_{T\rightarrow\infty}\frac{1}{T}\mathbb{E}\hat{\boldsymbol{u}}_{T}(\omega)\hat{\boldsymbol{u}}_{T}(\omega)^{*}d\omega
\]
It then holds that, for any measurable set $U\subset\mathbb{R}$, the matrix $C=\int_{U}d\hat{\gamma}_{\text{in}}\left(\omega\right)$ is positive semi-definite.  In particular, the decomposition yields $\gamma_{\text{in}}\left(0\right)=\int_{\mathbb{R}}d\hat{\gamma}_{\text{in}}\left(\omega\right)$ which is positive semi-definite by construction.  
Now, by linearity of the trace
\[
\left\Vert \gamma_{\text{in}}\left(0\right)\right\Vert _{1}
=\int_{\mathbb{R}}\left\Vert d\hat{\gamma}_{\text{in}}\left(\omega\right)\right\Vert _{1}
\] 
showing that the signal's power is additively distributed amongst its frequency components. 
We use the power spectral decomposition of $\gamma_{\text{in}}(\tau)$ to compute $\gamma_{\text{out}}(0)$. 
It gives
\[
\gamma_{\text{out}}\left(0\right)
= \int_{\mathbb{R}}\int_{0}^{\infty}\int_{0}^{\infty}e^{s_{1}(A-i\omega)}d\hat{\gamma}_{\text{in}}\left(\omega\right)e^{s_{2}(A-i\omega)^{*}}ds_{1}ds_{2}
\]
\[
= \int_{\mathbb{R}}\left(i\omega-A\right)^{-1}d\hat{\gamma}_{\text{in}}\left(\omega\right)\left(i\omega-A\right)^{-1*}
\] 
Using H\"older's inequality, we get
\[
\left\Vert \gamma_{\text{out}}\left(0\right)\right\Vert _{1}
\leq \int_{\mathbb{R}}\left\Vert \left(i\omega-A\right)^{-1}\right\Vert _{\infty}^{2}\left\Vert d\hat{\gamma}_{\text{in}}\left(\omega\right)\right\Vert _{1}
\]
\[
\;\;\;\;\;\;\;\;\;\;\;\;\;\;\;\;\;\;\;\;\;\;\;\;\;\;\;\;\;\;\;\;
\leq \sup_{\omega}\left\Vert \left(i\omega-A\right)^{-1}\right\Vert _{\infty}^{2}\left\Vert \gamma_{\text{in}}\left(0\right)\right\Vert _{1}
=\mathcal{V}_\text{D}^{2}\,\left\Vert \gamma_{\text{in}}\left(0\right)\right\Vert _{1}
\] 
with $\mathcal{V}_\text{D}=\sup_{\omega}\left\Vert \left(i\omega-A\right)^{-1}\right\Vert _{\infty}$
denoting deterministic variability. 

The inequality is strict for
$
\boldsymbol{u}\left(\lambda,t\right)=e^{i(\omega t-\lambda)}\boldsymbol{v}
$
with
$
\lambda\in\left(S^1,d\lambda\right)
$
where $d\lambda$ is the uniform measure on the circle, with $\omega$ and $\boldsymbol{v}\neq0$ satisfying
$
\left\Vert \left(i\omega-A\right)^{-1}\boldsymbol{v}\right\Vert =\mathcal{V}_\text{D}||\boldsymbol{v}||
$.  Indeed, notice that
$
d\hat{\gamma}_{\text{in}}(\omega)=\left(\boldsymbol{v}\boldsymbol{v}^{*}\right)\delta(\omega)d\omega
$
so that
$
\left\Vert \gamma_{\text{in}}\left(0\right)\right\Vert _{1}=||\boldsymbol{v}||^2
$.
Also
$
\gamma_{\text{out}}\left(0\right)=\left(i\omega-A\right)^{-1}\boldsymbol{v}\left(\left(i\omega-A\right)^{-1}\boldsymbol{v}\right)^{*}
$
so that
$
\left\Vert \gamma_{\text{out}}\left(0\right)\right\Vert _{1}=\left\Vert \left(i\omega-A\right)^{-1}\boldsymbol{v}\right\Vert _{\infty}^{2}=\mathcal{V}_\text{D}^{2}\,\left\Vert \gamma_{\text{in}}\left(0\right)\right\Vert _{1}
$.

\section{Harte's integrative measure of ecological stability}\label{sec:harte}

When defining stochastic variability section \ref{sec:stoch invariability}, to normalize the noise covariance matrix $\Sigma$ and to measure its effect on the system response $C_*$, we used the Frobenius norm $||\cdot||_\text{F}$. Other choices can be made, leading to slightly different results and interpretations. 
In this appendix we consider the trace norm 
\[
\left\Vert \Sigma\right\Vert _\text{Tr}=\mbox{Tr}\left(\sqrt{\Sigma^{\top}\Sigma}\right)
\]
which compares to the Frobenius norm as $\left\Vert \Sigma\right\Vert _\text{F} \leq \left\Vert \Sigma\right\Vert _\text{Tr} \leq \sqrt{S} \left\Vert \Sigma\right\Vert _\text{F}$ (recall that $S$ is the dimension of the system, e.g., number of species).  This choice leads to an interpretable notion of variability and facilitates the comparison between invariability and resilience.  Indeed, the trace norm of the system response $C_*$ is simply the expected square distance to equilibrium of the stationary distribution of $\boldsymbol{X}_*$,
\[
||C_*||_\text{Tr}=\mbox{Tr}\left(C_{*}\right)=\sum_{i=1}^{S}\mathbb{E}(X_{i}^{2})=\mathbb{E}(\left\Vert \boldsymbol{X}_*\right\Vert ^{2}).  
\]
For the trace norm, by convexity, the maximizing matrix $\Sigma$ is an orthogonal projector $\boldsymbol{u}\boldsymbol{u}^{\top}$ on a specific direction spanned by the vector $\boldsymbol{u}$, with $\left\Vert \boldsymbol{u}\right\Vert =1$.  One can then express the associated stationary covariance matrix as $C_{*}=\int_{0}^{\infty}e^{tA}\boldsymbol{u}(e^{tA}\boldsymbol{u})^{\top}dt$.  This leads to a different expression of intrinsic variability, namely (using linearity of the trace),
\begin{equation}
\mathcal{V}'_\text{S}=\sup_{\left\Vert \mathbf{u}\right\Vert =1}\int_{0}^{\infty}\left\Vert e^{tA}\boldsymbol{u}\right\Vert ^{2}dt\label{eq:VS-1}
\end{equation}
This definition of variability relates to the one derived using the Frobenius norm. In fact, it is rather straightfoward to show that the norm comparison is transported to the variability notions, giving 
\begin{equation}
\mathcal{V}_\text{S} \leq \mathcal{V}'_\text{S} \leq \sqrt{S} \mathcal{V}_\text{S}.\label{eq:VS1-VS2}
\end{equation}
At this point we can make an important remark on the link between intrinsic variability and resilience. Initial and asymptotic resilience are short- and long-term characteristic of the transient regime following a pulse perturbation. We see from (\ref{eq:VS-1}) that stochastic variability is related to the whole transient. 

In fact, Harte \citep{harte_ecosystem_1979} had proposed a stability measure $\mathcal{S}$, designed to integrate both short- and long-term responses of ecological communities. With our notations, for pulse perturbations, Harte's measure reads
\[
\mathcal{S}^{-1}=\sum_{i=1}^{S}\int_{0}^{\infty}|x_{i}(t)|^{2}dt=\int_{0}^{\infty}\left\Vert e^{tA}\boldsymbol{x}_{0}\right\Vert ^{2}\, dt
\]
Harte argued that this measure was empirically convenient, yet ``does not connect in any transparent way with methods of mathematical analysis''. To some extent, we have revealed this connection. The maximal value for $\mathcal{S}^{-1}$ over normalized pulse perturbations is exactly $\mathcal{V}'_\text{S}$, that is intrinsic stochastic variability, when defined with respect to the trace norm.

\section{Proof of the general stability ordering\label{sec:proof1}}

Let us here briefly sketch the proof of the chain of inequalities
(\ref{eq:dynamical_stab_ordering}) 
\[
\mathcal{R}_{0} \leq \mathcal{I}_\text{S} \leq \mathcal{I}_\text{D} \leq \mathcal{R}_{\infty}
\] 
Where $\mathcal{R}_X, \; \mathcal{I}_\text{Y}\;(X=0,\infty; \; \text{Y}=\text{S},\text{D})$ are the four intrinsic stability measures defined in the main text. We start from the classical inequality from pseudo-spectra analysis, giving a lower bound on the frequency response of the system $d\boldsymbol{x}/dt=A\boldsymbol{x}$ in terms of the excitation frequency $\omega$ and the dominant eigenvalue of the community matrix $A$: 
\begin{equation}
\left\Vert \left(i\omega-A\right)^{-1}\right\Vert \geq\left|i\omega-\lambda_{\mathrm{dom}}\right|^{-1}\label{eq:starting-point2}
\end{equation}
a proof of which can be found in the book by  \cite{trefethen_spectra_2005}. Another useful relation shows that resilience bounds the amplification envelope, in the sense that
\begin{equation}
e^{-\mathcal{R}_{\infty}t}\leq\left\Vert e^{tA}\right\Vert \leq e^{-\mathcal{R}_{0}t}\label{eq:starting-point}
\end{equation}
From the definition of deterministic variability (\ref{eq:variability}), the first expression (\ref{eq:starting-point2}) implies that 
$
\mathcal{V}_\text{D}\geq\left|\Re\left(\lambda_{\mathrm{dom}}\right)\right|^{-1}=\mathcal{R}_{\infty}^{-1}
$,
hence that $\mathcal{I}_\text{D}\leq\mathcal{R}_{\infty}$.  At this point, it is useful to give an alternative expression for the system's response direction $\boldsymbol{w}$ appearing in (\ref{eq:w}), namely: 
\[
\boldsymbol{w} = \int_{0}^{\infty}e^{t(A-i\omega)}\boldsymbol{u}\, dt 
\] 
Recall that the norm of this vector quantifies variability under deterministic forcing. By definition, we then have
\begin{equation}
\mathcal{V}_\text{D} \leq \max_{||\boldsymbol{u}||=1} \int_{0}^{\infty}\left\Vert e^{tA}\boldsymbol{u}\right\Vert \, dt \label{eq:VD}
\end{equation}
This shows that variability is bounded by the area under the amplification envelope $\left\Vert e^{tA}\right\Vert$, so that (\ref{eq:starting-point}) gives $\mathcal{R}_0 \leq \mathcal{I}_\text{D}$. 

We have showed that $\mathcal{R}_0 \leq \mathcal{I}_\text{D} \leq  \mathcal{R}_\infty$. We now prove  that
$
\mathcal{R}_{0} \leq \mathcal{I}_\text{S}
$.
We use the trace-normalization described in \ref{sec:harte}, to define intrinsic variability  (\ref{eq:VS-1}) and put  $\mathcal{I}'_\text{S}=\frac{1}{2} \mathcal{V'}_\text{S}^{-1}$. From (\ref{eq:VS1-VS2}) we have that  that $\mathcal{I}'_\text{S}\leq\mathcal{I}_\text{S}$. The expression (\ref{eq:VS-1}), along with (\ref{eq:starting-point}) above, gives the expected inequality. 

It thus remains to be proved that 
$
\mathcal{I}_\text{S}\leq\mathcal{I}_\text{D}                   
$. 
We shall need a lemma from linear algebra
\begin{lemma}
For any invertible matrix $B$ acting on $\mathbb{R}^N$, it holds that:
\[
\min_{x \in  \mathbb{R}^N ; \left\Vert x \right\Vert = 1} \left\Vert B x \right\Vert = ( \max_{y \in  \mathbb{R}^N ; \left\Vert y \right\Vert = 1} \left\Vert B^{-1} y  \right\Vert)^{-1}  
\]
\end{lemma}
\begin{proof}
Take $x_*=B^{-1}y/||B^{-1}y||$ with $y$ normalized and realizing the max of $\left\Vert B^{-1} y  \right\Vert $. By construction
$
\min_{x \in  \mathbb{R}^N ; \left\Vert x \right\Vert = 1} \left\Vert B x \right\Vert
\leq \left\Vert B x_* \right\Vert 
= ( \max_{y \in  \mathbb{R}^N ; \left\Vert y \right\Vert = 1} \left\Vert B^{-1} y  \right\Vert)^{-1}.
$ 
To show that taking the min over all normalized elements $x$ does not give anything smaller, it suffices to choose $y_*=Bx/||Bx||$ with $x$ normalized and realizing the min of $||Bx||$. By construction
$
\max_{y \in  \mathbb{R}^N ; \left\Vert y \right\Vert = 1} \left\Vert B^{-1} y  \right\Vert
\geq \left\Vert B^{-1} y_* \right\Vert 
= ( \min_{x \in  \mathbb{R}^N ; \left\Vert x \right\Vert = 1} \left\Vert B x  \right\Vert)^{-1} 
$, which is equivalent to
$
\min_{x \in  \mathbb{R}^N ; \left\Vert x \right\Vert = 1} \left\Vert B x \right\Vert
\geq ( \max_{y \in  \mathbb{R}^N ; \left\Vert y \right\Vert = 1} \left\Vert B^{-1} y  \right\Vert)^{-1};
$
proving the lemma.
\end{proof}
Now, with the above lemma, we get that
\[
2\mathcal{I}_{\text{S}}=(\sup_{\left\Vert \Sigma\right\Vert _{\mathrm{F}}=1}\left\Vert \hat{A}^{-1}\Sigma\right\Vert )^{-1}=\inf_{\left\Vert C\right\Vert _{\mathrm{F}}=1}\left\Vert \hat{A}C\right\Vert 
\]
and similarly
\[
\mathcal{I}_\text{D}=\inf_{\omega,\left\Vert \boldsymbol{v}\right\Vert =1}\left\Vert \left(i\omega-A\right)\boldsymbol{v}\right\Vert  
\]
Therefore, for any normalized matrix $C$, 
\[
2\mathcal{I}_{\text{S}}\leq\inf_{\left\Vert C\right\Vert _{\mathrm{F}}=1}\left\Vert \hat{A}C\right\Vert 
\]
If we choose $C$ as a rank-one orthonormal projector $C=\boldsymbol{v}\boldsymbol{v}^{*}$.
We then have that 
\[
2\mathcal{I}_{\text{S}}\leq \left\Vert \hat{A}C\right\Vert _\text{F}=\left\Vert \left(A\boldsymbol{v}\right)\boldsymbol{v}^{*}+\boldsymbol{v}\left(A\boldsymbol{v}\right)^{*}\right\Vert _\text{F}=\left\Vert \left(\left(i\omega-A\right)\boldsymbol{v}\right)\boldsymbol{v}^{*}+\boldsymbol{v}\left(\left(i\omega-A\right)\boldsymbol{v}\right)^{*}\right\Vert _\text{F}
\]
for any real $\omega$.
Choosing  $\boldsymbol{v}$ and $\omega$ such that $\mathcal{I}_\text{D}=\left\Vert \left(i\omega-A\right)\boldsymbol{v}\right\Vert$, yields 
\[
2\mathcal{I}_{\text{S}} \leq \left\Vert \hat{A}C\right\Vert _\text{F}\leq2\mathcal{I}_\text{D}.
\]
giving the full ordering (\ref{eq:dynamical_stab_ordering}).

\end{document}